\documentclass[preprint,12pt]{elsarticle}

\usepackage{amsmath,amssymb}
\usepackage{color}

\newtheorem{thm}{Theorem}

\newtheorem{assumption}{Assumption}
\newtheorem{defi}{Definition}

\newtheorem{remm}{Remark}

\newenvironment{proof}{\noindent {\em Proof.}}{\hfill \hspace*{1pt} \hfill $\square$}

\newcommand\real{\ensuremath{{\mathbb R}}}
\newcommand\realn{\ensuremath{{\mathbb{R}^n}}}

\newcommand{\calA}{\mathcal{A}}
\newcommand{\calB}{\mathcal{B}}
\newcommand{\calC}{\mathcal{C}}

\newcommand{\calR}{\mathcal{R}}

\newcommand{\calU}{\mathcal{U}}

\newcommand{\calK}{\mathcal{K}}

\newcommand{\calX}{\mathcal{X}}

\journal{Journal of Differential Equations}

\begin{document}

\begin{frontmatter}

\title{Differential positivity characterizes one-dimensional normally hyperbolic attractors}

\author[cambridge]{F. Forni\corref{corauthor}}
\ead{f.forni@eng.cam.ac.uk}
\author[liege]{A. Mauroy\fnref{grantAM}}
\ead{a.mauroy@ulg.ac.be}
\author[cambridge]{R. Sepulchre}
\ead{r.sepulchre@eng.cam.ac.uk}

\address[cambridge]{Department of Engineering, University of Cambridge, UK}
\address[liege]{Department of Electrical Engineering and Computer Science, University of Li{\`e}ge, BE}

\fntext[grantAM]{The paper presents research results of the Belgian Network DYSCO (Dynamical Systems, Control, and Optimization), funded by the Interuniversity Attraction Poles Programme, initiated by the Belgian State, Science Policy Office. The scientific responsibility rests with its authors. A. Mauroy is currently supported by a BELSPO (Belgian Science Policy) return grant.}

\cortext[corauthor]{Corresponding author.}

\begin{abstract}
The paper shows that normally hyperbolic one-dimensional compact attractors of smooth dynamical
systems are characterized by differential positivity, that is, the pointwise infinitesimal contraction of a smooth cone field.
The result is analog to the characterization of zero-dimensional hyperbolic attractors
by differential stability, which is the pointwise infinitesimal contraction of a Riemannian metric.
\end{abstract}


\end{frontmatter}

\section{Introduction}
\label{sec:introduction}
A linear operator $A$
is positive if it maps a cone $\calK$ into itself, i.e.
$A\mathcal{K} \subset \mathcal{K}$ \cite{Bushell1973}.
For linear dynamical systems $x^+ = A x$, $A:\realn \to \realn$,
positivity has the natural interpretation of invariance (and contraction, if the positivity is strict)
of the cone $\calK$ along the trajectories of the system.
For continuous systems 
$\dot{x} = A x$ positivity reads 
$e^{At} \calK \subset \calK$ for any $t > 0$.

Positivity has strong implications for the trajectories of the 
linear system \cite{Bushell1973}. 
Under irreducibility assumption, Perron-Frobenius theorem 
guarantees the existence of a dominant (largest) real eigenvalue for $A$
whose associated eigenvector - the Perron-Frobenius vector $\mathbf{w}$ -
is the unique eigenvector that belongs to the interior of $\mathcal{K}$.
As a consequence, the subspace spanned by $\mathbf{w}$
is an attractor for the linear system, that is, 
for any vector $x \in \mathcal{K}$, $x\neq 0$, 
\begin{equation}
\label{eq:power_iteration}
\lim_{n\to\infty} \frac{A^n x}{|A^n x|} = \mathbf{w} \ .
\end{equation}

A classical geometric interpretation of Perron-Frobenius theorem 
is the \emph{projective contraction} of 
linear positive systems \cite{Bushell1973,Birkhoff1957}:
the rays of the cone converge towardss each other along the system dynamics.
Positivity is at the core of a number of properties of
Markov chains, consensus algorithms and large-scale control
\cite{Bushell1973,Berman1994,Farina2000,Moreau2004,Sepulchre2010,Rantzer2015}.
A straightforward example in linear algebra is the convergence of 
the power iteration algorithm \cite[Chapter 7]{Strang2006}, directly expressed by \eqref{eq:power_iteration}.

\emph{Differential positivity} extends linear positivity
to the nonlinear setting. A nonlinear system $\dot{x} = f(x)$ (or a nonlinear iteration $x^+ = F(x)$)
is differentially positive
if its linearization along any given trajectory is positive. 
A detailed characterization is provided in Section \ref{sec:diff+}. 
The intuitive idea is that the linearized flow $\partial_x \psi(\cdot,x)$ 
along the trajectory $\psi(\cdot,x)$, $\psi(0,x)=x$, maps the cone (field) $\calK$ into itself.

Differentially positive systems generalize the important class
of monotone dynamical systems  \cite{Smith1995,Angeli2003,Hirsch2006},
which are differentially positive with respect to a \emph{constant}
cone field (on vector spaces).
Not surprisingly, differential positivity restricts
the asymptotic behavior of a nonlinear system:
Figure \ref{fig:fw} illustrates the dichotomic behavior of a differentially positive system.
Analog to the Perron-Frobenius vector in the linear case, the 
Perron-Frobenius vector field is 
an attractor for the linearized dynamics
(dashed line in Figure \ref{fig:fw}). The main results of 
differential positivity \cite{Forni2016} is that the trajectories of a 
differentially positive system either converge asymptotically to
the image of an integral curve of the Perron-Frobenius vector field,
or they move transversally to the Perron-Frobenius vector field. In the latter
case the Perron-Frobenius vector field defines a direction of maximal
sensitivity \cite{Forni2016}. 
\begin{figure}[htbp]
\centering
\includegraphics[width=0.36\columnwidth]{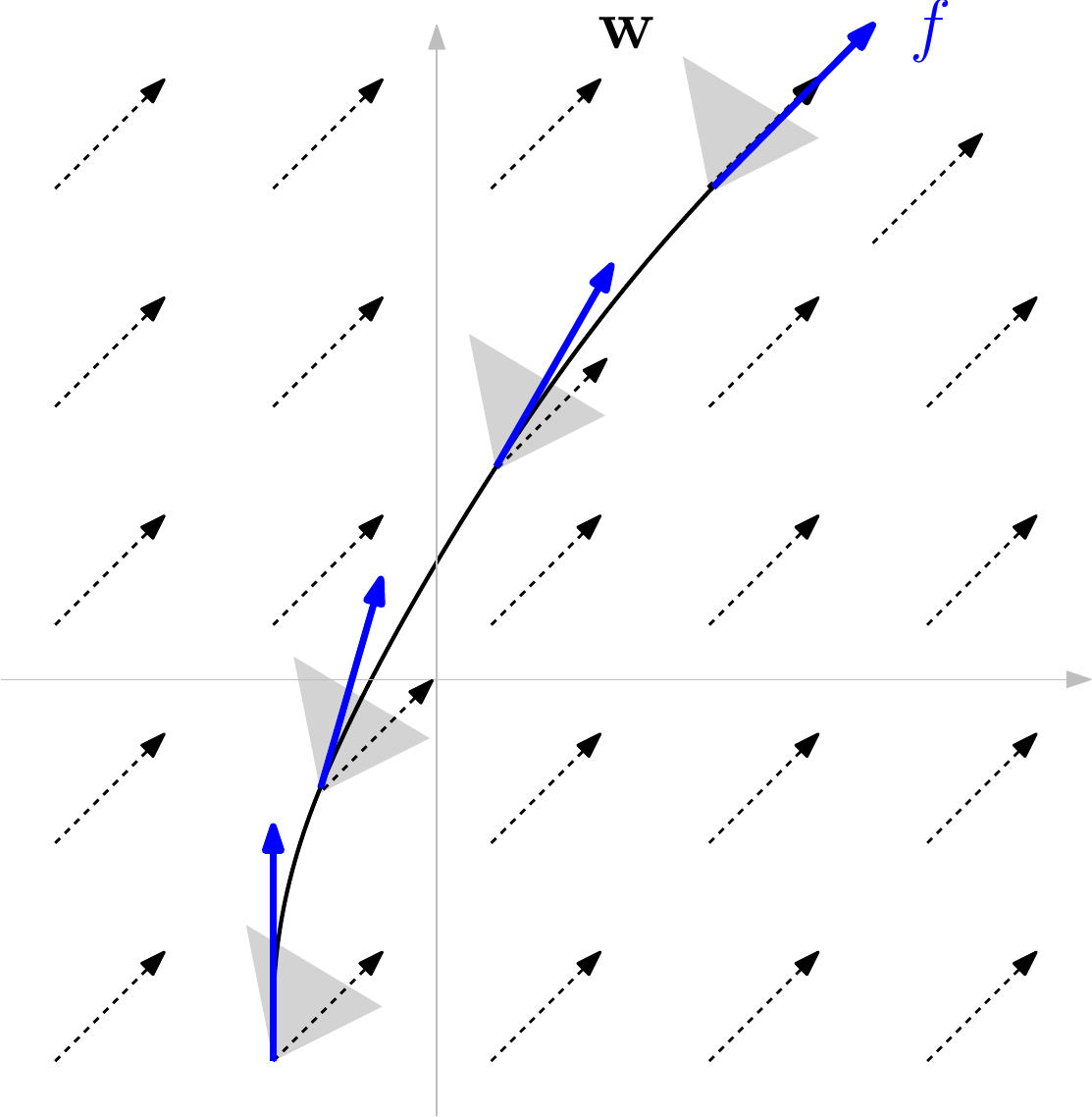}
\hspace{5mm}
\includegraphics[width=0.36\columnwidth]{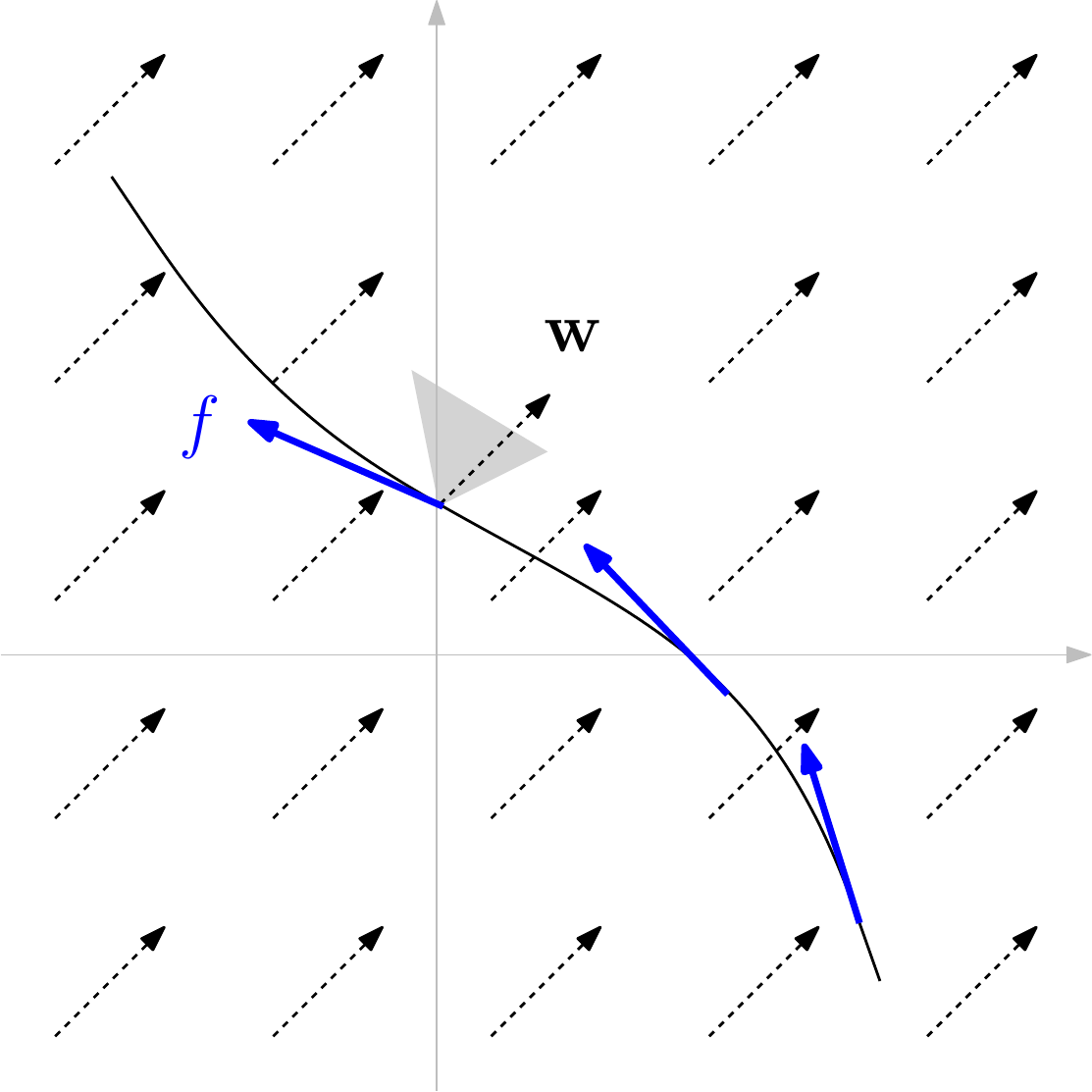}    
\caption{Left: trajectories converge to the image of an integral curve of the Perron-Frobenius vector
field $\mathbf{w}$. Right: trajectories move transversally to the Perron-Frobenius vector field $\mathbf{w}$,
which defines a direction of sensitivity.}
\label{fig:fw}
\end{figure}

In the former case of convergence, 
the image of an integral curve of the Perron-Frobenius vector field
defines a one-dimensional attractor for the nonlinear dynamics.
This attractor generalizes the subspace spanned by the Perron-Frobenius 
vector of linear positivity
(the Perron-Frobenius vector field is not constant in general).
It is either a collection of fixed points and 
connecting arcs or a limit cycle \cite{Forni2016,Forni2015}.
For this reason, differential positivity is a relevant analysis tool
for the study of bistable and periodic behaviors.

Further to the general exposition \cite{Forni2016}, the present paper strengthens 
the theory in the special case of \emph{hyperbolic} attractors. The main
result is the following converse theorem
\begin{thm}
\label{thm:converse_diff+}
A smooth dynamical system is (strictly) differentially positive in the basin of attraction of
a one-dimensional normally hyperbolic compact attractor.
\end{thm}

The paper is organized as follows.
Section \ref{sec:hyperbolic_fixed_points} 
recalls briefly the connections between
differential stability (contraction of
Riemannian metrics) and hyperbolic fixed points,
emphasizing the relevance of linearization methods
for the characterization of
asymptotic convergence to a fixed point.
A similar approach is pursued in
Sections \ref{sec:diff+}, \ref{sec:diff+_behavior} and 
\ref{sec:converse} to emphasize the analog role of
differential positivity in the characterization of 
asymptotic convergence to a one-dimensional attractor.
In particular, Section \ref{sec:diff+} introduces
the notion of differential positivity;
Section  \ref{sec:diff+_behavior} 
illustrates the main properties of differentially positive systems,
providing novel insights on the results in \cite{Forni2016};
Section \ref{sec:converse} completes and extends the results in \cite{Mauroy2015}
by proving Theorem \ref{thm:converse_diff+}.
Section \ref{sec:example} further
discusses the importance of the assumption of normal hyperbolicity.
Conclusions follow.

\section{Differential stability and hyperbolic fixed points}
\label{sec:hyperbolic_fixed_points}

\subsection{Differential stability (contraction)}
A linear system $\dot{x} = A x$, $x\in \realn$, is Lyapunov stable if for some positive 
definite matrix $P$ the Lyapunov function $V(x):= |x|_P^2 = x^TPx$ 
is non-increasing along the system trajectories
\begin{equation}
\dot{V} \ = \ x^T(A^T P + P A) x \ \leq \ 0 \qquad \forall x \in \realn \ .
\end{equation}
The strict inequality $\dot{V} < 0$ entails exponentially stability, that is,
the exponential contraction in time 
$e^{At}\calB_r \subset \calB_r$
of the ball $\calB_r := \{x\in\realn \,|\, |x|_P\leq r\}$ of radius $r>0$.

Analog to the linear case, a nonlinear system $\Sigma$ represented by 
$\dot{x} = f(x)$, $x\in \realn$, is \emph{differentially} exponentially stable (or contractive) 
if its linearization along any trajectory is exponentially stable.
This property has a pointwise geometric characterization based on 
the construction of Lyapunov-like functions for  
the prolonged system $\delta\Sigma$ \cite{Crouch1987} 
\begin{equation}
\label{eq:prolonged}
\left\{ 
\begin{array}{rcl}
\dot{x} &=& f(x) \\
\dot{\delta x} &=& \partial f(x) \delta x
\end{array} \quad\qquad (x,\delta x)\in \real^n\!\times\real^n \ .
\right.
\end{equation}
where $\partial f(x)$ is the differential of $f$ at $x$. 
Let $|\cdot|$ be any Riemannian metric. Consider
a function $V:\real^n\!\times\real^n \to \real_{\geq 0}$
such that 
\begin{itemize}
\item[(i)] there exist $0 < \lambda_{lb} <\lambda_{ub}  $, a positive integer $p$
\begin{equation}
\label{eq:Lyap}
\lambda_{lb} |\delta x|^p \leq V(x,\delta x) \leq \lambda_{ub} |\delta x|^p \qquad \forall (x,\delta x) \in \real^n\!\times\real^n \ ;
\end{equation}
\item[(ii)] there exists $\lambda >0$ 
\begin{equation}
\label{eq:dotLyap}
\dot{V}(x,\delta x) \leq -\lambda V(x,\delta x) \qquad \forall (x,\delta x) \in \real^n\!\times\real^n \ .
\end{equation}
\end{itemize}
Then, the nonlinear system is contractive in $\realn$ \cite[Theorem 1]{Forni2014}. The reader will notice that $V$ is just a Lyapunov function
lifted to the tangent bundle. In \eqref{eq:Lyap} the Riemannian metric $|\cdot|$ replaces the norms of classical Lyapunov
theory. The inequality \eqref{eq:dotLyap} characterizes the exponential stability of the linearization.
For instance, 
using $\psi:\mathbb{R}^+ \times \mathbb{R}^n \to \mathbb{R}^n$ 
to denote the nonlinear flow, we have that the pair 
$(\psi(\cdot,x), \partial_x \psi(\cdot,x) \delta x)$ is a trajectory of the prolonged system
from the initial condition $(x,\delta x)$.
Thus, 
\eqref{eq:Lyap} and \eqref{eq:dotLyap} guarantee the exponential contraction 
$\partial_x \psi(t,x)\calB_r \subset \calB_r$ of any
infinitesimal ball $\calB_r := \{\delta x \in \realn \,|\, |\delta x| \leq r \}$,
as illustrated in Figure \ref{fig:ball}.

The reader is referred to \cite{Forni2014} for a detailed 
Lyapunov characterization of contraction and for 
further extensions of the theory to general manifolds and Finsler metrics. 
Fundamental results of contraction theory and applications can be found in 
\cite{Lewis1949,Lohmiller1998,Rouchon2003,Pavlov05book,Russo10}.

\begin{figure}[htbp]
\centering
\includegraphics[width=0.7\columnwidth]{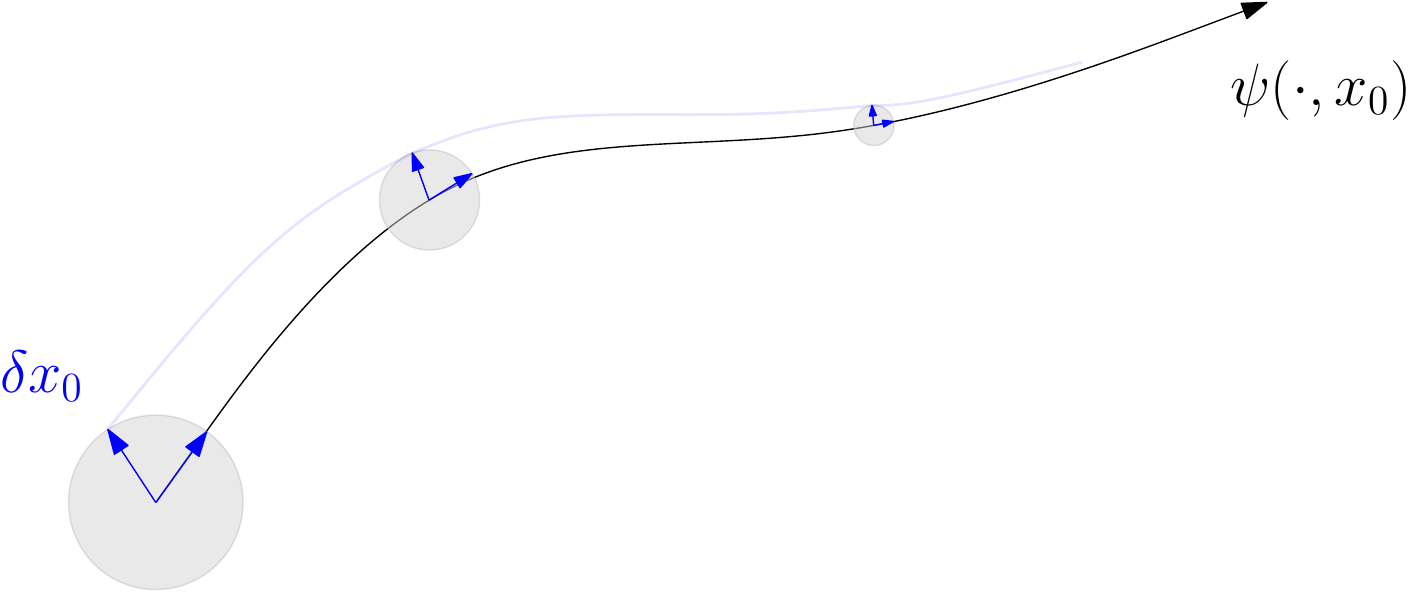}
\caption{Contraction of neighboring trajectories towardss $\psi(\cdot, x_0)$.}
\label{fig:ball}
\end{figure}

\subsection{The asymptotic behavior of differentially stable systems}

The exponential stability of the linearization \eqref{eq:Lyap}, \eqref{eq:dotLyap}
guarantees the contraction of the Riemannian metric along the system
trajectories over a uniform time horizon. As a consequence, 
there exists $K\geq 1$ such that 
\begin{equation}
\label{eq:dist_contr}
d(\psi(t,x),\psi(t,y)) \leq Ke^{-\lambda t} d(x,y)  \qquad \forall x,y\in \realn \,,\ \forall t\geq 0 \
\end{equation} 
where $d$ is the geodesic distance given by the Riemannian metric \cite[Theorem 1]{Forni2014},
as shown in Remark \ref{rem:contraction} for completeness.
Thus, $\psi(T,\cdot)$ is a contraction whenever $T\geq 0$ satisfies $Ke^{-\lambda T} < 1$. 
Using the semigroup property of the flow, by Banach fixed-point theorem
a contractive system has a unique exponentially stable fixed point $x^*$,
provided that the Riemannian metric $|\cdot|$ makes $\realn$ geodesically complete.

A simple direct argument exploits the identity
\begin{equation}
\label{eq:f}
 \frac{d}{dt} f(\psi(t,x)) = \partial f(\psi(t,x)) \frac{d}{dt}\psi(t,x) = \partial f(\psi(t,x)) f(\psi(t,x))
\end{equation} 
which holds along any trajectory $\psi(\cdot,x)$.
The identity shows that any pair $(\psi(\cdot,x),f(\psi(\cdot,x)))$ 
is a trajectory of the prolonged system \eqref{eq:prolonged}.
Therefore, by \eqref{eq:Lyap} and \eqref{eq:dotLyap},
every trajectory $\psi(\cdot,x)$ converges to a fixed point since 
\begin{equation} 
\label{eq:f_contr}
\lim\nolimits\limits_{t\to\infty} \left|\frac{d}{dt}\psi(t,x)\right|  =  \lim_{t\to\infty} |f(\psi(t,x))| =  0 \ .
\end{equation} 
Uniqueness follows from \eqref{eq:dist_contr}, by contradiction.
Furthermore, $x^*$ is exponentially stable by Lyapunov's first theorem, since  
\eqref{eq:Lyap} and \eqref{eq:dotLyap} hold at $x^*$.
The existence of a unique fixed point is indeed a necessary condition for 
contractive time-invariant systems. 
\begin{remm}
\label{rem:contraction}
Consider \eqref{eq:Lyap} and \eqref{eq:dotLyap}.
For any $x,y\in\realn$, 
consider the geodesic curve $\gamma(\cdot):[0,1] \to \calU$ 
connecting $\gamma(0) = x$ to $\gamma(1) = y$ with length 
$\ell(\gamma) := \int_0^1 \left|\frac{d}{ds}\gamma(s)\right| ds = d(x,y)$.
Then, the geodesic distance $d$ satisfies \eqref{eq:dist_contr} since
\begin{subequations}
\begin{align}
d(\psi(t,x),\psi(t,y))  
& \leq \! \int_0^1 \! \left|\frac{d}{ds} \psi(t,\gamma(s))\right | \!  ds   \nonumber  \\
&\leq \! \int_0^1 \!\! Ke^{-\lambda t} \!\left|\frac{d}{ds} \gamma(s)\right | \! ds = Ke^{-\lambda t} d(x,y)  
\end{align}
\end{subequations}
for some $K\geq 1$, 
where the second inequality follows from \eqref{eq:Lyap} and \eqref{eq:dotLyap}.
For instance, for any curve $\gamma(\cdot)$
\begin{subequations}
\begin{align}
\frac{d}{dt} \frac{d}{ds} \psi(t,\gamma(s))
&= \frac{d}{ds} \frac{d}{dt} \psi(t,\gamma(s)) \nonumber \\
&= \frac{d}{ds} f(\psi(t,\gamma(s)))  \nonumber \\
&= \partial f(\psi(t,\gamma(s))) \frac{d}{ds} \psi(t,\gamma(s))  \nonumber
\end{align}
\end{subequations}
which shows that $\frac{d}{ds} \psi(t,\gamma(s))$ is 
a trajectory of the linearized dynamics 
$\dot {\delta x} = \partial f(\psi(t,\gamma(s))) \delta x$
from the initial condition $\frac{d}{ds}\gamma(s)$.
Thus, by \eqref{eq:Lyap} and \eqref{eq:dotLyap}
\begin{equation}
\label{eq:delta_convergence}
\left |\frac{d}{ds} \psi(t,\gamma(s))\right| \leq K e^{-\lambda t} \left |\frac{d}{ds} \gamma(s)\right |
\end{equation}
for some $K \geq 1$.
\end{remm}

\subsection{Hyperbolic fixed points}

By Hartman-Grobman theorem, the trajectories 
of the system in a small neighborhood of a hyperbolic fixed point are topologically conjugated 
to the trajectories of the linearized dynamics \cite{Hartman1960}. 
Thus, at any equilibrium $0=f(x^*)$, the eigenvalues of the  
Jacobian matrix $\partial f(x^*)$ dictate the stability property of the fixed point. 
If $\partial f(x^*)$ is a Hurwitz matrix then $x^*$ is locally asymptotically stable.
A Hurwitz Jacobian matrix also guarantees that the nonlinear system
is contractive in a small neighborhood $\calU$ of the fixed point: by continuity, there
exists $P>0$ such that 
$\partial f(\psi(t,x))^T P + P \,\partial f(\psi(t,x)) < 0$ 
holds for each $x$ in a small neighborhood $\calU$
of the fixed point. Thus, \eqref{eq:Lyap} and \eqref{eq:dotLyap} hold in $\calU$
for $V := \delta x^T P \delta x$.

As a matter of facts, 
the existence of a unique hyperbolic fixed point is also a sufficient condition
for contraction in the whole $\realn$, as summarized in the next theorem.
A detailed analysis can be found in \cite{Giesl2015}. 
We provide a new proof based on \cite[Theorem 2.3]{LanMezicJan2013}. 

\begin{thm}
A smooth dynamical system
is contractive in the basin of attraction
of a hyperbolic fixed point.
\end{thm}
\begin{proof}
Without loss of generality, consider a hyperbolic fixed point
$x^* = 0$ and let $\cal{B}$ be the basin of attraction of $x^*$.
By \cite[Theorem 2.3]{LanMezicJan2013} 
there exists a differentiable $h:\calB \to \realn$ such that
$y = h(x)$ is a $C^1$ diffeomorphism with 
$\partial h(0) = I$ and $\dot{y} = \partial f(0) y =: A y$.
By hyperbolicity, there exists a symmetric and positive definite
matrix $P$ such that $A^T P + P A \leq -\lambda P$
for some $\lambda > 0$.
 
In the $y$ coordinates, consider now the Riemannian metric $|\cdot|_P$
defined by $|\delta y|_P := \sqrt{\delta y^T P \delta y}$ at each $y$.  
\eqref{eq:Lyap} and \eqref{eq:dotLyap} are trivially satisfied by 
selecting $V := \delta y^T P \delta y$. 
From \cite[Theorem 1]{Forni2014}
the nonlinear system is contractive in
the basin of attraction $h(\calB)=\mathbb{R}^n$.
In the original coordinates,
the metric is represented by 
$ |\delta x|_P := \sqrt{\delta x^T \partial h(x)^T P \partial h(x) \delta x}$ 
at each $x$ and
$V := \delta x^T \partial h(x)^T P \partial h(x) \delta x$.
\end{proof}

\section{Differential positivity}
\label{sec:diff+}

Differential positivity replaces the contraction 
of a metric, the shrinking ball in Figure \ref{fig:ball},
with the contraction of a cone, illustrated 
in Figure \ref{fig:cone}. Projective contraction
- the contraction of the rays of the cone -
frees one direction of the contraction property.
The result is that the trajectories are not 
forced to converge to a fixed point anymore.
The unique fixed point is replaced by one-dimensional 
attractors, the image of a curve.
\begin{figure}[htbp]
\centering
\includegraphics[width=0.7\columnwidth]{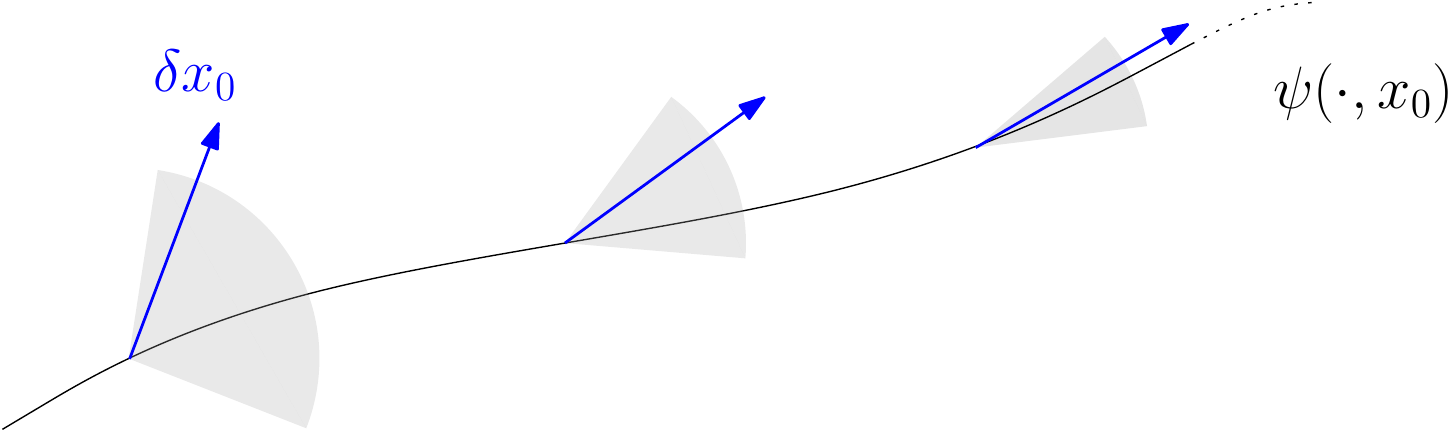}
\caption{Contraction of a cone along  $\psi(\cdot, x_0)$.}
\label{fig:cone}
\end{figure}

Differential positivity is defined below
for complete continuous 
systems $\dot{x} = f(x)$,
whose state belongs to a smooth $n$-dimensional 
Riemannian manifold $\calX$. 
In such a case the state $(x,\delta x)$ of the prolonged system
\eqref{eq:prolonged} belongs to the tangent bundle $T\calX$.
The manifold $\calX$ is endowed with a cone field
\begin{equation*}
\mathcal{K}(x) \subseteq T_x \mathcal{X} \qquad \forall x \in \mathcal{X} 
\end{equation*}
which makes $\calX$ a conal manifold \cite{Lawson1989,Neeb1993}.
Each cone $\mathcal{K}(x)$ is closed 
and solid, and satisfies the following 
properties: for all $x \in \mathcal{X}$, 
(i) $\mathcal{K}(x)+\mathcal{K}(x) \subseteq \mathcal{K}(x)$, 
(ii) $\alpha \mathcal{K}(x) \subseteq \mathcal{K}(x) $ for all $\alpha \in \mathbb{R}^+$, (iii) $\mathcal{K}(x) \cap -\mathcal{K}(x)=\{0\}$
(i.e. convex and pointed). 
We assume that $\calK$ is continuous and piecewise differentiable, i.e.
the boundary of $\calK$ has a continuous and piecewise differentiable
parameterization.

Positivity of the linearization has the precise interpretation 
of forward invariance and contraction 
of the cone field along the prolonged dynamics 
$\delta \Sigma$ \cite{Forni2016},
as clarified in the following definition.
For simplicity,
in what follows we use 
$\psi^t(\cdot):\calX\to\calX$ to denote
$\psi^t(x) := \psi(t,x)$ where $\psi(\cdot,x) \in \Sigma$.
\begin{defi}[differential positivity]
\label{def:diff_pos}
The system $\Sigma$ is differentially positive 
(with respect to the cone field $\mathcal{K}$) in $\calX$ 
if the flow of the prolonged system \eqref{eq:prolonged}
leaves the cone invariant
\begin{equation}
\label{eq:diff+}
\partial \psi^t(x) \mathcal{K}(x) \subseteq \mathcal{K}(\psi^t(x)) \qquad \forall x\in \mathcal{X}\,, \forall t>0\,.
\end{equation}
In addition, a differentially positive system $\Sigma$ is (uniformly) strictly differentially positive if there exist a constant $T>0$ and a cone field $\mathcal{R}(x) \subset \textrm{int} \mathcal{K}(x) \cup \{0\}$ such that
\begin{equation}
\label{eq:sdiff+}
\partial \psi^t(x) \mathcal{K}(x) \subseteq \mathcal{R}(\psi^t(x)) \qquad \forall x\in \mathcal{X}\,, \forall t \geq T\,. 
\end{equation} 
\end{defi}
To avoid pathological cases, we assume that for every pair of points
$x_1,x_2\in \calX$, 
there exists a linear invertible mapping  
$\Gamma(x_1,x_2):T_{x_1}\calX \to T_{x_2}\calX$ 
such that 
$\Gamma(x_1,x_2)\calK(x_1) = \calK(x_2)$
and
$\Gamma(x_1,x_2)\calR(x_1) = \calR(x_2)$.

Strict differential positivity 
guarantees that every tangent vector $\delta x$ 
on the boundary of $\calK(x)$ is mapped into 
$\partial \psi^t(x) \delta x \in \calR(\psi^t(x))$
at time $t\geq T$.
Combining this property with the linearity of the mapping 
leads to the projective contraction in Figure \ref{fig:cone}.
The result has a metric characterization 
based on a generalization of the argument
in \cite{Bushell1973}. In particular, 
\cite{Bushell1973} shows that linear positive maps 
contract the Hilbert metric (projective metric).
\cite{Bushell1973} reduces 
the existence of the Perron-Frobenius eigenvector 
for linear positive maps to an application of the
Banach fixed point theorem. 
Differentially positive systems enjoy a similar property.

From \cite[Section VI]{Forni2016},
for any given $x\in \calX$,
take any $\delta x, \delta y \in \calK(x)\setminus\{0\}$ and define 
$
M_{\calK(x)}(\delta x,\delta y) \!:=\!
 \inf \{\lambda \!\geq\! 0 \,:\, \lambda \delta y - \delta x \!\in\! \calK(x) \} 
$ \footnote{
$M_{\calK(x)}(\delta x,\delta y) := \infty$
when $ \{\lambda \in \real_{\geq 0} \,:\, \lambda \delta y - \delta x \in \calK(x) \} = \emptyset$.
}
and
$
m_{\calK(x)}(\delta x, \delta y) := 
 \sup \{\lambda \geq 0 \,:\, \delta x - \lambda \delta y  \in \calK(x) \} 
$.
The Hilbert metric $h_{\calK(x)}$ reads
\begin{equation}
\label{eq:lifted_hilbert_metric}
 h_{\calK(x)}( \delta x, \delta y) := \log \left(\frac{ M_{\calK(x)}(\delta x,\delta y)}{ m_{\calK(x)}(\delta x,\delta y)}\right) \ .
\end{equation}

$h_{\calK(x)}$ is a projective metric
that measures the distance between rays of the cone.
In particular, 
for any $\delta x,\delta y \in \calK(x)$,
$h_{\calK(x)}(\delta x,\delta y) = 0$
if and only if $\delta x=\lambda \delta y$ with $\lambda \geq 0$,
and $h_{\calK(x)}(\alpha \delta x,\beta \delta y)=
h_{\calK(x)}(\delta x,\delta y)$ for any $\alpha>0$ and $\beta>0$.
The Hilbert metric $h_{\calK(x)}$ reduces to a metric in 
$\calK(x) \cap \{\delta x \in T_x\calX\,:\, |\delta x|_x \!= \!1\} $.
We make the following standing assumption.
\begin{assumption}
\label{assume:completeness}
For all $x\in \calX$, the pair
$\calK(x)\cap \{\delta x \in T_x\calX\,:\, |\delta x| =1\}$ and 
$h_{\calK(x)}$ is a complete metric space. 
\end{assumption}
We recall that $|\cdot|$ is the Riemannian metric on $\calX$.
Examples of complete metric spaces on cones
can be found in \cite{Bushell1973,Lemmens2012,Zhai2011}.

The contraction of the cone field along trajectories
is captured by the exponential convergence of the Hilbert metric.
Following \cite[Theorem 2]{Forni2016},
for a strictly differentially positive
system there exist $\rho\geq1$ and $\lambda > 0$ 
such that, for all $x\in \calX$, $\delta x_1,\delta x_2 \in \calK(x)$, 
and $t\geq T$, 
\begin{equation}
\label{eq:Hilbert_contraction}
 h_{\calK(\psi^t(x))}
 (\partial \psi^t(x) \delta x_1,\partial \psi^t(x) \delta x_2) 
 \leq \rho e^{-\lambda(t-T)} \Delta 
\end{equation}
where $\Delta := 
\sup \{ h_{\calK(x)}(v_1,v_2) \,:\, v_1,v_2 \in \calR(x) \} < \infty$.

Assumption \ref{assume:completeness} and the 
projective contraction \eqref{eq:Hilbert_contraction}
lead to the existence of the so-called 
Perron-Frobenius vector field $\mathbf{w}(x) \in \textrm{int}\calK(x)$,
the differential equivalent of the Perron-Frobenius eigenvector
of linear positive mappings. 
By \cite[Theorem 3]{Forni2016},
for any $x\in \calX$, $\mathbf{w}$ is a continuous vector field such that
\begin{equation*}
\label{eq:def_PF_vec}
\mathbf{w}(x) := 
\lim_{t \rightarrow \infty} \frac{\partial \psi^t(\psi^{-t}(x)) \delta x}{|\partial \psi^t(\psi^{-t}(x)) \delta x |}
\qquad
\mbox{where }\delta x \in \mathcal{K}(\psi^{-t}(x)) \setminus \{0\} \ .
\end{equation*}
$\mathbf{w}$ is invariant as a field of rays, 
that is, 
$\mathbf{w}(\psi^t(x)) =  \partial\psi ^t(x) \mathbf{w}(x) / |\partial\psi ^t(x) \mathbf{w}(x)|$.
Furthermore, 
for any $x\in \calX$ and 
$\delta x \in \calK(x)$, 
\begin{equation}
\label{eq:PF_convergence}
\lim_{t\to\infty}
h_{\calK(\psi^t(x))}(\mathbf{w}(\psi^t(x)), \partial \psi^t(x)\delta x) = 0 \ .
\end{equation}

The observation in \eqref{eq:f} makes clear that 
\eqref{eq:PF_convergence} is at the root of the dichotomic behavior illustrated in Figure \ref{fig:fw}.
In brief, for any trajectory $x(\cdot)$, either $\dot{x}(t) \in \calK(x(t))$ for some $t$,
which forces the trajectory to converge asymptotically towardss an integral curve of the Perron-Frobenius
vector field, or $\dot{x}(t) \notin \calK(x(t))$ for all $t$, which determines the transversality
between $\dot{x}(t)$  and $\mathbf{w}(x(t))$, leading to sensitivity. 

In what follows, for simplicity, we call \emph{Perron-Frobenius curve} any integral curve
of the Perron-Frobenius vector field.

\begin{remm}
Differential positivity makes contact with the 
Alekseev and Moser criterion which 
infers the hyperbolicity of 
an attractor $\calA$ from the existence of 
two invariant cone fields defined at each $x\in \calA$.
\cite[Chapter 3]{HandbookDynamicalSystemsV1A}.
Differential positivity 
also requires the existence of an invariant cone field.
Differential positivity exploits the 
contraction of the cone field
everywhere in the system state manifold
(or in any forward invariant subset) 
to characterize the global asymptotic behavior of the system.
In this sense, differential positivity shows similarities
with \cite{Wojtkowski1985}, which
uses the invariance of cone fields on $\calX$
to characterize the Lyapunov exponents of the system.
\end{remm}

\section{The asymptotic behavior of a differentially positive system}
\label{sec:diff+_behavior}

In order to make the paper self-contained, we summarize
in this section the main results of \cite{Forni2016},
which characterize the asymptotic behavior of differentially positive systems.

The $\omega$-limit set $\omega(\xi)$, $\xi \in \mathcal{X}$, is the set
\begin{equation*}
\omega(\xi) := \bigcap_{T \in \mathbb{R}} \overline{\{\psi^t(\xi),t>T\}}\ ,
\end{equation*}
where $\overline{\phantom{X}}$ denotes the closure of the set. For every $\xi\in \calX$,
the $\omega$-limit set
of a (complete) strictly differentially positive system 
whose trajectories are all bounded
satisfies one of the following two properties \cite[Theorem 4]{Forni2016}:
\begin{itemize}
\item[(i)] The vector field $f(x)$ is aligned with the Perron-Frobenius vector field $\mathbf{w}(x)$
for each $x\in \omega(\xi)$ (i.e. $f(x)=\lambda(x) \mathbf{w}(x)$, $\lambda(x) \in \real$), and $\omega(\xi)$ is either a fixed point 
or a limit cycle or a set of fixed points and connecting arcs; 
\item[(ii)] The vector field $f(x)$ is not aligned with the 
Perron-Frobenius vector field $\mathbf{w}(x)$ for each $x \in \omega(\xi)$ 
such that $f(x)\neq 0$, 
and either 
$
\liminf\nolimits\limits_{t\to\infty} 
|\partial \psi^t(x)\mathbf{w}(x)| = \infty
$
or 
$\lim\nolimits\limits_{t\to\infty} f(\psi(t,x)) = 0$.
\end{itemize}

\subsection{Convergence to a one-dimensional attractor}
The behavior (i) is akin to Poincar{\'e}-Bendixson theorem 
for planar systems \cite{Hirsch1974}. It holds if 
the Perron-Frobenius vector field is complete and
if the following holds:
\begin{equation}
\label{eq:PF_boundedness}
\limsup_{t \to \infty} |\partial \psi^t(x)\mathbf{w}(x)| < \infty 
\qquad\quad  \forall x\in \calC \ .
\end{equation}
In such a case the attractor is given by the image of a Perron-Frobenius curve.
The combination of \eqref{eq:PF_boundedness} with the projective contraction 
of strict differential positivity
leads to contraction transversally to the Perron-Frobenius vector field:
$n-1$ directions of contraction for an $n$-dimensional state manifold $\calX$. 
As a consequence, every pair of trajectories in $\calC$ converge asymptotically
to the image of a unique Perron-Frobenius curve. The completeness of
the vector field ensures existence and uniqueness of such a curve. 
A detailed argument is developed in Remark \ref{rem:horizontal_contraction}.

\begin{remm}
\label{rem:horizontal_contraction}
We first establish contraction transversal to the Perron-Frobenius vector field:
for every $x\in \calC$ and $\delta x \in T_x \calX$, 
\begin{itemize}
\item[(a)]
$\limsup\nolimits\limits_{t\to\infty}|\partial \psi^t(x)\delta x|< \infty$; and
\item[(b1)] either  
$\lim\nolimits\limits_{t\to\infty}|\partial \psi^t(x)\delta x| = 0$
\item[(b2)] or
$\lim\nolimits\limits_{t\to\infty} h_{\calK(\psi^t(x))} \left(q \partial \psi^t(x) \delta x, \mathbf{w}(\psi^t(x))\right) = 0$
 for some $q\in \{1,-1\}$.
\end{itemize}
If $\delta x \in -\calK(x) \cup \calK(x)$ 
then (a) follows from the fact that $\mathbf{w}(\psi^t(x))$
is a dominant direction and $\partial \psi^t(x)$ is a linear operator. 
(b2) follows by projective contraction and linearity of $\partial \psi^t(x)$. 
(b1) may also occur.
For $\delta x \notin -\calK(x) \cup \calK(x)$  either there exists $\tau>0$ such that 
$\partial \psi^\tau(x) \delta x \in -\calK(\psi^\tau(x))\cup \calK(\psi^\tau(x))$ and (a),(b2) 
follow by the argument above, and (b1) may occur, or 
$\partial \psi^t(x) \delta x \notin \calK(\psi^t(x))$ for all $t \geq 0$.
For this last case, 
take $\delta y := \alpha \mathbf{w}(x) + \delta x$. For $\alpha>0$ sufficiently
large $\delta y \in \calK(x)$ and (a),(b2) hold for $\delta y\in \calK(x)$. 
Therefore $\lim\nolimits\limits_{t\to\infty} | \partial\psi^t(x)\delta x| = 0$,
as claimed in (b1) (and (b1) implies (a)).

We now look at the convergence among trajectories by 
parameterizing their initial conditions with a curve $\gamma(\cdot):[0,1] \to \calC$. 
For each $s\in [0,1]$, $\psi^t(\gamma(s))$ 
converge asymptotically to the image of a Perron-Frobenius curve
as $t\to\infty$.  This follows from the observation that for each $s\in[0,1]$ 
$\frac{d}{ds}\psi^t(\gamma(s))=\partial \psi^t(\gamma(s)) \frac{d}{ds} \gamma(s)$
and the pair $\gamma(s) \in \calC$, $\frac{d}{ds} \gamma(s) \in T_{\gamma(s)}\calX$ satisfies
(a),(b2) or (b1) along the flow of the system.
The Perron-Frobenius limit curve must be unique. 
Otherwise, there exists a curve $\gamma(\cdot):[0,1] \to \calC$
connecting several Perron-Frobenius limit curves
such that, for some $s\in [0,1]$, the tangent vector $\frac{d}{ds} \gamma(s)$
along the linearized flow
$\frac{d}{ds}\psi^t(\gamma(s)) = \partial \psi^t(\gamma(s)) \frac{d}{ds} \gamma(s)$ 
does not satisfy one of (a) and (b2), and it does not satisfy (b1).
\end{remm}

Various assumptions may ensure the ultimate boundedness property \eqref{eq:PF_boundedness}.
For instance, \cite[Corollary 2]{Forni2016}
ensures the existence of a unique attractive limit cycle
in $\calC$ under the simpler assumption 
\begin{equation}
\label{eq:limit_cycles}
f(x) \in \mathit{int}\calK(x) \qquad \quad \forall x\in \calC \ .
\end{equation}
\eqref{eq:limit_cycles} guarantees
that the trajectories do not converge to a fixed point, 
since the magnitude of the vector field is bounded from 
below $|f(x)| > b > 0$ on any orbit (by boundedness of trajectories). 
\eqref{eq:limit_cycles} also implies \eqref{eq:PF_boundedness} (see Lemma 1 in \cite{Forni2016}) 
thus transversal contraction with respect to the Perron-Frobenius vector field. 
Finally, \eqref{eq:limit_cycles} guarantees that 
every $\omega$-limit set in $\calC$ satisfies (i). 
The vector field $f(x)$ at each point $x$ of the
$\omega$-limit set is aligned with 
the Perron-Frobenius vector field, therefore every trajectory
in the neighborhood of the $\omega$-limit set is attracted
transversally towards the $\omega$-limit set. 
As a consequence, the return map on a Poincar{\'e} section transversal to
the $\omega$-limit set is necessarily contractive, leading
to the existence of a closed orbit. The reader is referred to
\cite[Corollary 2]{Forni2016} and \cite[Theorem 6]{Forni2015} for
a detailed discussion.

The asymptotic stability of the limit cycle readily follows
from the analysis of the linearization. Let $x$ be a point on the closed orbit
and let $\tau>0$ be the period such that $\psi^\tau(x) = x$. By strict differential 
positivity \footnote{Necessarily $T$ in \eqref{eq:sdiff+} satisfies $T \leq \tau$.},
\begin{equation*}
\partial \psi^\tau(x) \calK(x) \subseteq \mathit{int}\calK(\psi^\tau(x))\cup\{0\} = \mathit{int}\calK(x)\cup\{0\} \ .
\end{equation*}
$\partial \psi^\tau(x)$ is thus a positive linear operator 
with dominant eigenvector $\mathbf{w}(x) = \frac{f(x)}{|f(x)|}$ since, by periodicity 
of $\partial \psi^\tau(x)$ and using \eqref{eq:f}, 
$\partial \psi^\tau(x) f(x) = f(\psi^\tau(x)) = f(x)$. The dominant eigenvalue clearly has magnitude 1.
By strict differential positivity the other eigenvalues - the Floquet's characteristic multipliers
of the system - have magnitude less than one. It follows that the closed orbit
is asymptotically stable  \cite[Chapter 13]{Hirsch1974}.

\subsection{Sensitive behaviors}

(ii) characterizes $\omega$-limit sets of points $x\in \calX$
for which \eqref{eq:PF_boundedness} does not hold. In such a case, 
the Perron-Frobenius vector field is a direction of maximal sensitivity.
Fixed points or periodic orbits are unstable because of
the asymptotic separations among trajectories in a small neighborhood 
of the $\omega$-limit set, in the direction of the Perron-Frobenius vector field.
However, strict differential positivity still allows for one-dimensional
attractors defined by a collection of fixed points and connecting arcs.
The simplest example is given by a bistable system, with 
two stable fixed points and a saddle point $x^*$. 
If the linearization of the system at the saddle
has $n-1$ eigenvalues with negative real part, then
strict differential positivity guarantees that all the fixed points 
belong to the image of a Perron-Frobenius curve.
Indeed, the whole heteroclinic orbit connecting the 
unstable manifold of the saddle to any stable fixed points 
is contained within the image of a Perron-Frobenius curve.
This follows from the continuity of the Perron-Frobenius vector field,
combining the fact that the tangent space of the unstable manifold at $x^*$ is spanned by $ \mathbf{w}(x^*)$
with the invariance of the Perron-Frobenius vector field
$\mathbf{w}(\psi^t(x)) =  \partial\psi ^t(x) \mathbf{w}(x) / |\partial\psi ^t(x) \mathbf{w}(x)|$
and with the invariance of the unstable manifold.

Further results related to case (ii) can be found in 
\cite{Forni2016} and \cite{Forni2015}. Notably, 
the relevant property of almost global convergence 
of monotone systems is revisited through differential positivity.
In general, however, the relation between differential positivity and nonlinear 
behavior for case (ii) requires further investigation. 
It is an open question, for example, whether or not a 
strictly differentially positive system may have strange attractors.

\section{Normally hyperbolic attractors}
\label{sec:converse}

At a fixed point $x^*$, strict differential positivity reduces to classical positivity of the linearization 
since $\partial \psi^t(x^*) \calK(x^*) \subseteq \mathrm{int}\calK(x^*)$ for all $t > 0$.
In such a case, $\mathbf{w}(x^*)$ is the dominant eigenvector of 
$\partial f(x^*)$ and the stability of the fixed point is determined by the associated eigenvalue $\lambda_{\mathbf{w}(x^*)}$.
$\mathbf{w}(x^*)$ can be an unstable direction - $\liminf\nolimits\limits_{t\to\infty} 
|\partial \psi^t(x^*)\mathbf{w}(x^*)| = \infty$ of (ii) - 
which shows that differential positivity does not imply stability.
On the other hand, the tangent space splits into
\begin{equation}
	T_{x^*}\calX = W(x^*) \oplus N(x^*)
\end{equation}
where $W(x^*) := \{\lambda \mathbf{w}(x^*)\,:\, \lambda \in\real\}$
and $N (x^*)$ is given by the span of the remaining eigenvectors of
$\partial f(x^*)$. The linearized flow contracts $N(x^*)$ 
more sharply than $W(x^*)$ or it expands $W(x^*)$ more sharply than $N(x^*)$, that is, 
for some  $c(x^*)>0$,
\begin{equation*}
	\log\left(\frac{|\partial \psi^t(x^*) \mathbf{w}(x^*)|}{|\partial \psi^t(x^*) \delta x|}\right) \geq c(x^*) t
	\qquad\quad \forall \delta x \in N(x^*) , |\delta x| = 1 \ .
\end{equation*}
A separation of subspaces and rates is akin to normal hyperbolicity \cite{Hirsch1977}.
\begin{remm}
On any fixed point $x^*$, strict differential positivity guarantees that $ e^{\partial f(x^*) t} \mathbf{w}(x^*) \in W(x^*)$
which makes $\mathbf{w}(x^*)$ an eigenvector of  $\partial f(x^*)$.
The real part of an eigenvalue $\lambda_v$ associated to any other eigenvector $v$ of $\partial f(x^*)$
must satisfy $\lambda_v < \lambda_{\mathbf{w}(x^*)}$, otherwise \eqref{eq:PF_convergence}
would not hold. $c(x^*)$ is thus any positive constant such that $\lambda_{\mathbf{w}(x^*)} - \lambda_v \geq c(x^*) > 0 $.
\end{remm}

We need the following definitions and assumptions.

\begin{defi}
\label{def:normal_attractor}
A connected manifold $\mathcal{A} \subseteq \mathcal{X}$ 
 is an attractor for the flow $\psi$ if it satisfies the following two properties:
\begin{itemize}
\item $\psi^t(\mathcal{A}) \subseteq \mathcal{A}$ for all $t>0$;
\item There exists a neighborhood $\mathcal{U} \subset \mathcal{X}$ of $\mathcal{A}$ such that,
for any $x\in \calU$, the limit set $\omega(x)$ is in $\mathcal{A}$.
\end{itemize}
\end{defi}

Definition \ref{def:normal_attractor}
allows for the existence of a smaller closed subset satisfying 
the above properties. For example, the definition includes the case 
of a set of fixed points and their connecting orbits. 

In what follows we will use the notion of invariant 
splitting of the tangent bundle of $\mathcal{X}$ over $\mathcal{A}$,
represented by 
\begin{equation*}
T_x \mathcal{X} = N_x \oplus T_x \mathcal{A} \qquad \quad \forall x \in \calA \ ,
\end{equation*}
with
$
\partial \psi^t(x) T_x \mathcal{A}  = T_{\psi^t(x)} \mathcal{A}
$
and 
$
\partial \psi^t(x) N_x = N_{\psi^t(x)} 
$
for all $x\in \calA$ and $t\geq 0$.

\begin{defi}
\label{def:hyperbolic}
The attractor $\mathcal{A}$ is normally hyperbolic (with respect to the flow $\psi$) if there exists an invariant 
splitting of the tangent bundle at $\mathcal{A}$ which satisfies the following property:
there exist a Riemannian metric $|\cdot |$
and real constants $
\rho_1 \geq 1$, $\rho_2 \geq 1$, $\lambda_1>0$, $\lambda_2<\lambda_1$ such that,
for all $x\in \mathcal{A}$ and $t \geq 0$, 
\begin{subequations}
\label{eq:hyperbolic_splitting}
\begin{align}
\label{eq:hyperbolic_normal}
| \partial \psi^t(x) \delta x | & \leq   \rho_1 e^{-\lambda_1 t} |\delta x| \qquad \forall \delta x \in N_x \\
\label{eq:hyperbolic_tangent}
| \partial \psi^t(x) \delta x | & \geq   \rho_2 e^{-\lambda_2 t} |\delta x| \qquad \forall \delta x \in T_x \mathcal{A} \ .
\end{align}
\end{subequations}
\end{defi}

We can now precisely state and prove the main result of the paper.
\begin{thm}
\label{thm:converse_diff+_detailed}
Consider a dynamical system $\Sigma$ on a Riemannian manifold $\calX$, 
represented by $\dot{x}=f(x)$ where $x\in \calX$ and $f$ is a $C^1$ vector field.
Let  $\mathcal{A} \subset \mathcal{X}$ be a normally hyperbolic compact attractor for $\Sigma$ with basin of attraction 
$\mathcal{B}_{\calA}$.
If $\dim T_x \mathcal{A} = 1$ for all $x\in \mathcal{A}$ then 
$\Sigma$ is strictly differentially positive in $\mathcal{B}_\calA$.
\end{thm}

\begin{proof}

\underline{Strict differential positivity in $\calA$}.
According to \cite[Remark 3]{Hirsch1977} the normal hyperbolicity property with $\dim T_x \mathcal{A} = 1$ implies that one can define an adapted metric $|\cdot|^*$ in $T \mathcal{X}$, $x\in \mathcal{A}$, such that \eqref{eq:hyperbolic_normal}-\eqref{eq:hyperbolic_tangent} are satisfied with $\rho_1=\rho_2=1$, i.e. for all $x \in \mathcal{A}$, there exist real constants $\lambda_1>0$ and $\lambda_2<\lambda_1$ such that \footnote{
It follows from \eqref{eq:hyperbolic_normal} that one can define an adapted metric in $N_x$ for all $x \in \mathcal{A}$, i.e.
$| \partial \psi^t(x) v |^*  \leq  e^{-\lambda_1 t} |v|^*$ for all $v \in N_x$ (see e.g. \cite{Hirsch1969}). This corresponds to \eqref{hyperbolic_normal2}. In addition, since $\dim T_x \mathcal{A} = 1$, there exists a function $\Lambda$ such that $\partial \psi^t(x) v = e^{-\Lambda(t) t} v$ for all $t>0$ and $v \in  T_x \mathcal{A}$, and \eqref{eq:hyperbolic_tangent} implies $\limsup_{t\rightarrow \infty} \Lambda(t) \leq \lambda_2$. For $v \in T_x \mathcal{A}$, we can define the metric
\begin{equation*}
|v|^*=\int_0^T |{\partial}\psi^t(x)v| e^{\lambda_2 t} \,dt
\end{equation*}
where $T$ is large enough. We have
\begin{equation*}
\left. \frac{d|\psi^\tau(x)v|^*}{d\tau}\right|_{\tau=0} = \frac{d}{d\tau} \left[e^{-\lambda_2 \tau} \int_\tau^{T+\tau} |{\partial}\psi^t(x)v| e^{\lambda_2 t} dt \right]_{\tau=0} = -\lambda_2 |v|^* + |v| (e^{(\lambda_2-\Lambda(T))T}-1)
\end{equation*}
where the second term is positive for $T$ large enough. Then we obtain \eqref{hyperbolic_tangent2}.
}
\begin{subequations}
\begin{align}
| \partial \psi^t(x) v |^* & \,\leq\,  e^{-\lambda_1 t} |v|^* 
&&  \forall v \in N_x  \label{hyperbolic_normal2} \\
| \partial \psi^t(x) v |^* & \,\geq\, e^{-\lambda_2 t} |v|^* 
&&  \forall v \in T_x \mathcal{A} \ . \label{hyperbolic_tangent2}
\end{align}
\end{subequations}
Consider the continuous vector field $\xi: \mathcal{A} \to T \mathcal{A}$ such that $|\xi(x)|^*  =  1$ for all $x \in \mathcal{A}$ and
\begin{equation*}
T_x \mathcal{A} = \{\alpha \xi(x) : \alpha \in \mathbb{R} \} \qquad x \in \mathcal{A}\,.
\end{equation*}
Given any $0<\varepsilon \ll 1$, 
at each point $x\in \calA$, we define the cone fields
\begin{subequations}
\label{eq:KandR}
\begin{align}
\mathcal{K}(x) & \, := \,
 \{\alpha \xi(x) + v \,:\, \alpha - |v|^* \geq 0, \alpha \in \mathbb{R}^+, v \in N_x\} 
 \\
\mathcal{R}(x) & \, := \,
\{\alpha \xi(x) + v \,:\, \alpha - |v|^* \geq \alpha\varepsilon , \alpha \in \mathbb{R}^+, v \in N_x\} \ .
\end{align}
\end{subequations}
Directly from the 
definition $\calK(x)$ is solid and closed
at each $x\in \calA$. $\calK(x)$ is a convex cone
since for any $\alpha \xi(x) + v \in \mathcal{K}(x)$ and $\alpha' \xi(x) + v' \in \mathcal{K}(x)$, we have $v+v'\in N_x$ and $\alpha+\alpha' \geq |v|^*+|v'|^* \geq |v+v'|^*$, so that $\mathcal{K}(x)+\mathcal{K}(x) \subseteq \mathcal{K}(x)$. Clearly $\gamma \mathcal{K}(x) \subseteq \mathcal{K}(x)$ for $\gamma>0$. Finally, if $\alpha \xi(x)+v \in \{\mathcal{K}(x),-\mathcal{K}(x)\}$, then $\alpha - |v|^* \geq 0$ and $-\alpha - |v|^* \geq 0$, so that $\alpha=|v|^*=0$. Thus $\mathcal{K}(x) \cap -\mathcal{K}(x) = \{0\}$, which makes $\calK(x)$ pointed.
The same holds for $\calR(x)$.

We show strict differential positivity.
For all $t>0$, $x\in \calA$ and $v \in N_x$, we have $\partial \psi^t(x)v \in N_{\psi^t(x)}$ and $\partial \psi^t(x) (\alpha \xi(x)) = (\alpha |\partial \psi^t(x) \xi(x)|^*)\, \xi(\psi^t(x))  \in T_{\psi^t(x)} \mathcal{A}$
since the splitting is invariant. We obtain
\begin{equation*}
\alpha | \partial \psi^t(x) \xi(x)|^* - |\partial \psi^t(x) v| 
\geq \alpha e^{-\lambda_2 t} - e^{-\lambda_1 t} |v|^* \geq 0
\end{equation*}
for all $t \geq 0$ since $\alpha \geq |v|^*$ and $\lambda_2<\lambda_1$. 
\eqref{eq:diff+} follows.
Furthermore, uniformly on $x$, there exists
$T_1 > 0$ such that $e^{(\lambda_2-\lambda_1)t} \leq (1-\epsilon)$ for all $t \geq T_1$, so that
\begin{equation*}
 (1-\varepsilon) \alpha e^{-\lambda_2 t} - e^{-\lambda_1 t} |v|^* \geq e^{-\lambda_2 t} (1-\epsilon) (\alpha-|v|^*) \geq 0
\end{equation*}
for all $t\geq T_1$, which implies \eqref{eq:sdiff+}.

\underline{Strict differential positivity in a small neighborhood $\calU\supset\calA$}. For any positive constant $c$, denote by $\calU_c$ the (compact) set of points 
$y\in\mathcal{B}_{\mathcal{A}}$ whose distance from $\mathcal{A}$ is less 
than or equal to $c$. 
By the hyperbolicity of $\mathcal{A}$, 
for any given $c >0$ sufficiently small, 
there exists a constant $0 < \rho(c) < c$ such that 
$\psi^t(y) \in \mathcal{U}_{c}$ for all $t \geq 0$ and $y\in \mathcal{U}_{\rho(c)}$.

Since $\mathcal{A}$ is normally hyperbolic and compact, there exists an invariant fibration of its stable manifold, i.e. of the basin of attraction $\mathcal{B}_{\mathcal{A}}$, in the compact set $\calU_c$ (see Theorem 4.1 in \cite{Hirsch1977} and Theorem 2 in \cite{Fenichel1974}).
Denote by $\mathcal{W}(x) \subset \calU_c$ the fiber characterized by a section $x \in \mathcal{A}$ and
denote by $N_y$ the tangent space of the fiber $\mathcal{W}(x)$ at $y \in \mathcal{W}(x)$.
Note that 
the invariance property implies $\partial\psi^t(y)N_y\subseteq N_{\psi^t(y)}$.

Define a continuous extension $\overline{\xi}$ of the vector field 
$\xi$ on $\calU_c$ such that for all $y\in \calU_c\setminus\calA$
\begin{itemize}
\item
$|\overline{\xi}(y)|^* = 1$; 
\item
$\overline{\xi}(y) \in T_y \calX \setminus N_y$;
\item
whenever $\psi^t(y)\in \calU_c$, 
there exists $\rho$ such that
$\partial\psi^t(y)\overline{\xi}(y) = \rho \overline{\xi}(\psi^t(y))$.
\end{itemize}
At each point $y\in \calU_c$ define the cone fields $\calK$ and $\calR$
as in \eqref{eq:KandR}.

We show strict differential positivity on 
$\calU := \bigcup_{t\geq 0} \psi^t(\calU_{\rho(c)}) \supset \calA$,
for $c>0$ sufficiently small.
By continuity of the vector field $\xi$ on the compact set $\calU_c$
and 
by continuity of the prolonged flow \eqref{eq:prolonged}
with respect to initial conditions,
for $c$ sufficiently small, there exists a bound
$0< L(c) \to 0$ as $c \to 0$
such that,
\begin{subequations}
\label{eq:L} 
\begin{align}
|\partial \psi^t(y) v|^*\,  
&\leq \, e^{(-\lambda_1 + L(c)) t}  |v|^* 
&& \forall y \in \calU\,,\forall v \in N_y
\label{eq:L1} \\
|\partial \psi^t(y) \overline{\xi}(y)|^*\,  
&\geq \, e^{(-\lambda_2 - L(c)) t}  
&& \forall y \in \calU \ .
\label{eq:L2} 
\end{align}
\end{subequations}
The bound $L(c)$ arises from
(i) the loss of contraction rate of $\partial \psi^t(y)$ 
in the direction $v \in N_y$
between 
$y \in \mathcal{U} \setminus \mathcal{A}$ 
and $y \in \mathcal{A}$; 
(ii) the variation of the expansion/contraction rate of $\partial \psi^t(y)$ 
in the direction $\overline{\xi}(y)$
between 
$y \in \mathcal{U} \setminus \mathcal{A}$ 
and $y \in \mathcal{A}$.

For any 
$y \in \mathcal{U}$ and any 
$\delta y \in \mathcal{K}(y)$,
we have 
$\delta y = \alpha \overline{\xi}(y) + v $
where $v\in N_y$ and $\alpha - |v|^* \geq 0$. 
Then, by \eqref{eq:L}, for all $t\geq 0$ 
\begin{equation*}
\alpha |\partial \psi^t(y) \overline{\xi}(y)|^* -
|\partial \psi^t(y) v|^* 
\geq
\alpha e^{(-\lambda_2 - L(c)) t}  
- e^{(-\lambda_1 + L(c)) t}  |v|^* \geq 0
\end{equation*}
since $L(c)\to 0$ as $c\to 0$ and therefore
$\lambda_2< \lambda_1$ implies
$\lambda_2 +L(c) < \lambda_1-L(c)$
for $c$ sufficiently small. It follows that \eqref{eq:diff+} holds. 
For strict differential positivity, as before, there exists
$T_2 > 0$ such that, for all $t \geq T_2$,
\begin{equation}
\label{eq:sdiff+_calU}
(1-\varepsilon)\alpha e^{(-\lambda_2 - L(c)) t}  - e^{(-\lambda_1 + L(c)) t}  |v|^* \geq  e^{-\lambda_2 t} (1-\epsilon) (\alpha-|v|^*) \geq 0.
\end{equation}
Note that $T_2 \rightarrow T_1$ as $c\rightarrow 0$.

\underline{Strict differential positivity in $\calB_{\calA} \setminus \calU$}.
Consider the boundary $\mathit{bd}\calU$ of $\calU$. 
Following \eqref{eq:KandR}, for each $x\in \mathit{bd}\calU$,
define the $\rho$-parametrized cone fields
\begin{subequations}
\begin{align*}
\mathcal{K}^\rho(x) 
&= \{\alpha \overline{\xi}(x) + v \,:\, \alpha - e^\rho |v|^* \geq 0, 
\alpha \in \mathbb{R}^+, v \in N_x\} \\
\mathcal{R}^\rho(x) 
&= \{\alpha \overline{\xi}(x) + v \,:\, \alpha - e^\rho |v|^* \geq \varepsilon \alpha, 
\alpha \in \mathbb{R}^+, v \in N_x\} \ .
\end{align*}
\end{subequations}
Clearly,
$\mathcal{K}^{\rho'}(x)  \subset  \mathcal{K}^{\rho}(x)$ 
for any  $\rho'>\rho$,
since $\alpha - e^{\rho'} |v|^* \geq 0$ implies $\alpha - e^\rho |v|^* > 0$.
Furthermore, 
$\calR^\rho = \calK^{\rho+\overline{\varepsilon}}$
for $\overline{\varepsilon} := -\log (1-\varepsilon)$.
In fact, 
$(1-\varepsilon)\alpha - e^\rho |v|^* \geq 0$ 
if and only if
$\alpha - \frac{e^\rho}{(1-\varepsilon)} |v|^* \geq 0$,
and
$\frac{e^\rho}{(1-\varepsilon)} = e^{\rho+\overline{\varepsilon}}$.

We are ready to define the cone fields 
on $\mathcal{B}_{\mathcal{A}}\setminus \mathcal{U}$.
For $x \in \mathcal{B}_{\mathcal{A}} \setminus \mathcal{U}$, 
take
\begin{equation*}
\tau(x) :=  \min \{t \in \mathbb{R}_{\geq 0}\,:\, \exists x_0\in \calU, \psi(-t,x_0)=x \} 
\ 
\end{equation*} 
and let $x_0(x) \in \mathit{bd}\calU$ be the (unique) initial condition 
that satisfies the identity
$x=\psi(-\tau(x),x_0(x))$. Define
\begin{subequations}
\label{eq:cones_Ba_minus_U}
\begin{align}
\mathcal{K}(x) 
&:= \partial \psi(-\tau(x),x_0(x)) \mathcal{K}^{\tau(x)}(x_0(x)) 
&&x \in \mathcal{B}_{\mathcal{A}}\setminus \mathcal{U} \\
\mathcal{R}(x) 
&:= \partial \psi(-\tau(x),x_0(x)) \mathcal{R}^{\tau(x)}(x_0(x)) 
&&x \in \mathcal{B}_{\mathcal{A}}\setminus \mathcal{U} \ .
\end{align}
\end{subequations}
$\calK$ and $\calR$ are well defined since
the class of cones is closed under the action of linear maps. 
$\calK$ and $\calR$ are continuous by construction.

We prove differential positivity.
Using \eqref{eq:cones_Ba_minus_U}, for all $x\in \mathcal{B}_{\mathcal{A}}\setminus \mathcal{U}$ 
and  $t \leq \tau(x)$, we have
\begin{equation}
\label{eq:Ba_inclusion1}
\begin{split}
\partial \psi^t(x) \mathcal{K}(x) & = \partial \psi^{t-\tau(x)}(x_0(x)) \mathcal{K}^{\tau(x)}(x_0(x)) \\
& \subseteq \partial \psi^{t-\tau(x)}(x_0(x)) \mathcal{K}^{\tau(x)-t}(x_0(x)) \\
& = \partial \psi^{-\tau(\psi^t(x))}(x_0(\psi^t(x))) \mathcal{K}^{\tau(\psi^t(x))}(x_0(\psi^t(x))) \\
& =  \mathcal{K}(\psi^t(x))
\end{split}
\end{equation}
where we have used the identities
$t-\tau(x)=-\tau(\psi^t(x))$ and  $x_0(x)=x_0(\psi^t(x))$. 
The inclusion follows from the property
$\mathcal{K}^{\rho'}(x)  \subset  \mathcal{K}^{\rho}(x)$ 
for any  $\rho'>\rho$.  
For all $x\in \mathcal{B}_{\mathcal{A}}\setminus \mathcal{U}$ 
and  $t > \tau(x)$, 
we have 
\begin{equation}
\label{eq:Ba_inclusion2}
\begin{split}
\partial \psi^t(x) \mathcal{K}(x) 
& = \partial \psi^{t-\tau(x)}(x_0(x)) \mathcal{K}^{\tau(x)}(x_0(x)) \\
& \subseteq \partial \psi^{t-\tau(x)}(x_0(x)) \mathcal{K}^0(x_0(x)) \\
& \subseteq \mathcal{K}(\psi^{t-\tau(x)}(x_0(x))) \\
& = \mathcal{K}(\psi^t(x)) \ .
\end{split}
\end{equation}
The second inclusion follows from 
differential positivity in $\calU$. Differential positivity is thus established.

For strict differential positivity, by repeating the argument above, 
note that  \eqref{eq:Ba_inclusion1} and \eqref{eq:Ba_inclusion2}
hold when $\calK$ is replaced by $\calR$. Thus, we have only to
show that there exists a uniform time $T$ for which
$\partial \psi^T(x) \calK(x) \subseteq \calR(x)$.
For instance, 
consider the case 
$\overline{\varepsilon}\leq \tau(x)$,
where $\overline{\varepsilon} := -\log (1-\varepsilon)$.
Exploiting the identity
$ \mathcal{K}^{\rho}(x_0(x)) = \mathcal{R}^{\rho-\overline{\varepsilon}}(x_0(x))$, we have
\begin{equation}
\label{eq:strict_contract}
\begin{split}
\partial \psi^{\overline{\varepsilon}(x)} \mathcal{K}(x) 
& = \partial \psi^{\overline{\varepsilon}-\tau(x)}(x_0(x)) \mathcal{K}^{\tau(x)}(x_0(x)) \\
& = \partial \psi^{\overline{\varepsilon}-\tau(x)}(x_0(x)) \mathcal{R}^{\tau(x)-\overline{\varepsilon}}(x_0(x)) \\
& = \partial \psi^{-\tau(\psi^{\overline{\varepsilon}}(x))}(x_0(\psi^{\varepsilon}(x))) \mathcal{R}^{\tau(\psi^{\overline{\varepsilon}}(x))}(x_0(\psi^{\overline{\varepsilon}}(x))) \\
& =  \mathcal{R}(\psi^{\overline{\varepsilon}}(x)) \ .
\end{split}
\end{equation}
For $\overline{\varepsilon} > \tau(x)$,
$\psi^{t}(x) \in \mathrm{bd}\calU$ for some $t \in [0,\overline{\varepsilon}]$.
Thus strict differential positivity follows from \eqref{eq:strict_contract} and \eqref{eq:sdiff+_calU}, for a uniform $T=T_2+\overline{\varepsilon}$.
\end{proof}
\begin{remm}
The contracting cone field $\mathcal{K}(x)$ proposed in the proof is not unique. 
Several definitions can be provided, starting from the 
parameterization $\mathcal{K}^{\rho}(x)$ whose constant $e^\rho$ 
can be replaced by any function $k(\rho,x)$, $x\in \mathit{bd}\calU$, such that 
$k(0,x) = 1$ and $k(\cdot,x)$ is strictly increasing for each $x$. 
Looking at Section \ref{sec:diff+}, this implies that, in principle, different Perron-Frobenius vector fields
could arise from different cone fields. In general, 
the Perron-Frobenius vector field is uniquely defined only at $x\in\calA$.

The contracting cone field $\mathcal{K}(x)$ is continuous and piecewise differentiable on $\mathcal{B}_{\mathcal{A}}$.
In fact, it is not $C^1$ on the boundary of $\calU$. However, we suspect that one can obtain a contracting cone field that 
is $C^1$ everywhere in $\mathcal{B}_{\mathcal{A}}$ by choosing a properly parametrized cone field $\mathcal{K}^{\rho}(x)$.
\end{remm}

Theorem \ref{thm:converse_diff+_detailed} extends and completes the work initiated in 
\cite{Mauroy2015} on converse theorems for
normally hyperbolic attractors given by a unique fixed point or by a limit cycle. 
The argument for these converse results is based on the properties
of the Koopman operator \cite{Budisic2012}. 
In particular, the proof derives and exploits an interesting connection between 
the Koopman eigenfunctions related to the attractor \cite{Mauroy2013a,Mauroy2013} 
and the existence of an invariant (and contracting) cone field.
The proof is constructive, leading to 
a numerical algorithm for the construction
of the cone field \cite[Section 5]{Mauroy2015}.
However, the approach in \cite{Mauroy2015} does not extend to
the general one-dimensional normally hyperbolic attractors considered in Theorem \ref{thm:converse_diff+_detailed}.
For attractors containing several fixed points,
a cone field from Koopman eigenfunctions is well-defined in
the basin of attraction of each fixed point, but the patching 
of these cone fields in the whole basin
of attraction of the one-dimensional attractor is not necessarily well defined.

\section{Why normal hyperbolicity?}
\label{sec:example}
The relevance of the normal hyperbolicity property for Theorem \ref{thm:converse_diff+_detailed}
is readily illustrated by a comparison of the attractors in Figure \ref{fig:example}.
The one-dimensional attractor $\calA_\ell$ on the left of the figure is given
by two stable fixed points, a saddle point $x^*$, and the heteroclinic
orbits connecting the fixed points to the saddle - the unstable
manifold of the saddle.
The one-dimensional attractor $\calA_r$ on the right of the figure  is
characterized by a stable fixed point, a saddle point $x^*$, the
heteroclinic orbit connecting the fixed point to the saddle,
and the homoclinic orbit connecting
unstable and stable manifolds of the saddle. 
We assume that each fixed point is hyperbolic.
For simplicity we take $\calX := \real^2$ and 
$\calB_\calA = \calX$. 
\begin{figure}[htbp]
\centering
\includegraphics[width=0.66\columnwidth]{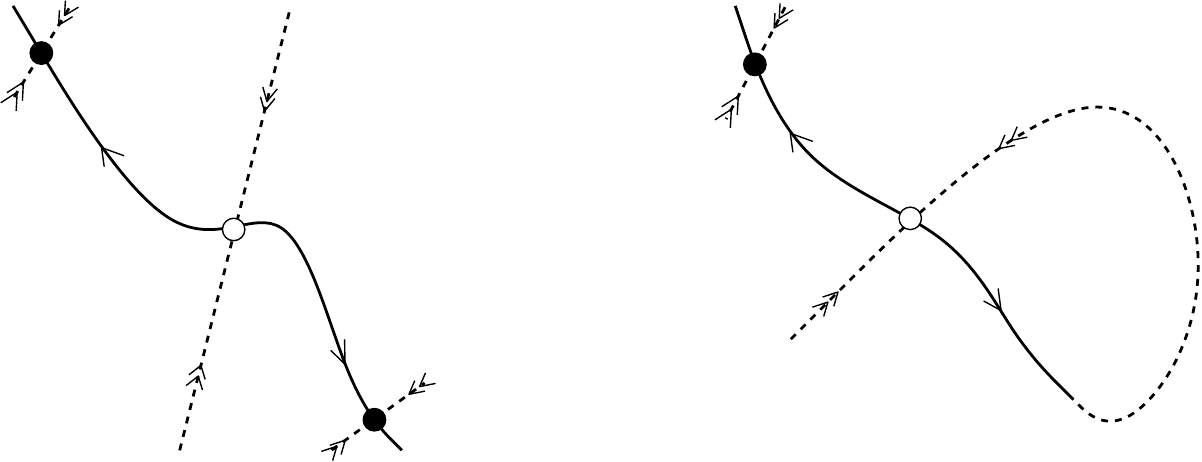}
\caption{One-dimensional attractors defined
by a set of fixed points and connecting arcs.}
\label{fig:example}
\end{figure}

$\calA_\ell$ is normally hyperbolic. The 
transversal convergence towardss the attractor
is dominated by the motion within the attractor. 
This attractor is compatible with differential positivity. 
Indeed, by Theorem \ref{thm:converse_diff+},
the attractor $\calA_\ell$ makes
the system strictly differentially positive in $\calX$.

$\calA_r$ is not normally hyperbolic
since unstable and stable manifolds of the
saddle merge. Thus, there is no invariant splitting of
the tangent bundle at $\calA_r$. In fact,
the system cannot be strictly differentially positive. We show this
by contradiction. Suppose that the system is
strictly differentially positive with respect to
some cone field $\calK$.
The Perron-Frobenius vector field $\mathbf{w}$
is a continuous vector field that satisfies 
$\mathbf{w}(x) \in T_x \calA$ at each $x \in \calA$. 
However, for any point $x$ on the homoclinic orbit we have that
$\lim\nolimits\limits_{t\to\infty}\mathbf{w}(\psi^t(x)) 
\neq \lim\nolimits\limits_{t\to \infty}\mathbf{w}(\psi^{-t}(x))$
despite 
$\lim\nolimits\limits_{t\to\infty} \psi^t(x) = 
\lim\nolimits\limits_{t\to\infty} \psi^{-t}(x) = x^*$.
This contradicts the continuity of $\mathbf{w}$.

A similar argument can be developed 
from the invariance of the cone field $\calK$.
By invariance, $\calK(x^*)$ must contain
both the tangent vectors 
$\delta x_s$ and $\delta x_u$
respectively tangent to the stable and unstable manifolds
of the saddle. 
For instance, (i) $\delta x_{u}$ is the dominant eigenvector
of the linearized flow on $x^*$ therefore $\delta x_u\in \calK(x^*)$;
(ii) by continuity, 
$\lim\nolimits\limits_{t\to\infty} \calK(\psi^{-t}(x)) = \calK(x^*)$
for any $x$ on the homoclinic orbit. Thus,
the tangent vector $\delta x$ to the homoclinic orbit at $x$ belongs to $\calK(x)$
and, by differential positivity
$\lim\nolimits\limits_{t\to\infty} \partial\psi^t(x)\calK(x) \subseteq \calK(x^*)$.
Therefore, $\delta x_s \in \calK(x^*)$.
However, considering the eigenvalues $\lambda_1 > \lambda_2$ of the linearization at the saddle $x^*$,
we have $\partial\psi^t(x^*) \delta x_s = e^{-\lambda_1 t} \delta x_s$
and
$\partial\psi^t(x^*) \delta x_u = e^{-\lambda_2 t} \delta x_u$
which contradicts the projective contraction property 
\eqref{eq:Hilbert_contraction}, since
$\lim\nolimits\limits_{t\to\infty}
d_{\calK(x^*)}(\partial\psi^t(x^*) \delta x_s,\partial\psi^t(x^*) \delta x_u)
= d_{\calK(x^*)}(\delta x_s,\delta x_u) > 0$.

\section{Conclusions}
\label{sec:conclusions}

Differential analysis, or the study of the nonlinear map through the properties
of its linearization, is a classical topic of dynamical systems theory.
We have shown that behaviors that asymptotically converge to a 
one-dimensional normally hyperbolic attractor can be characterized differentially:
their linearizations are \emph{positive}, they contract a cone field. 
This characterization can be thought as analog to the differential characterization of
behaviors that asymptotically converge to a hyperbolic fixed point; those behaviors
that contract a ``ball field'', that is, a Riemannian metric.
The result once more stresses the importance of hyperbolicity in differential analysis. 
It also emphasizes a fruitful connection between the property of positivity in
linear analysis and a geometric characterization of dynamical systems
with one-dimensional asymptotic behaviors.

\end{document}